\documentclass[10pt,a4paper]{article}
\usepackage[utf8]{inputenc}
\usepackage[english]{babel}
\usepackage{amsmath}
\usepackage{amsfonts}
\usepackage{amssymb}
\usepackage{amsthm}
\usepackage{graphicx}
\usepackage[left=2cm,right=2cm,top=2cm,bottom=2cm]{geometry}
\newtheorem{te}{Theorem}

\newtheorem{de}{Definition}

\providecommand{\keywords}[1]
{\small	\textbf{\textit{Keywords:}} #1}
\author{R. Azuaje and
A. M. Escobar-Ruiz\\
Departamento de F\'{i}sica, Universidad Aut\'onoma Metropolitana Unidad Iztapalapa\\ San Rafael Atlixco 186, 09340 Cd. Mx., M\'exico}
\title{Canonical and canonoid transformations for Hamiltonian systems on (co)symplectic and (co)contact manifolds}
\begin{document}
\maketitle

\begin{abstract}

In this paper we present canonical and canonoid transformations considered as global geometrical objects for Hamiltonian systems. Under the mathematical formalisms of symplectic, cosymplectic, contact and cocontact geometry, the canonoid transformations are defined for (co)symplectic, (co)contact Hamiltonian systems, respectively. The local characterizations of these transformations is derived explicitly and it is demonstrated that for a given canonoid transformation there exist constants of motion associated with it. 


\end{abstract}

\section{Introduction}
\label{sec1}

In classical mechanics the theory of coordinate transformations is fundamental and appears practically in any textbook\cite{Landau}. In particular, within the Hamiltonian symplectic mechanics the well known canonical transformations play an important role; if for a given Hamiltonian system one can find a canonical transformation such that half of the new coordinates are constants of motion, then the integration of the corresponding Hamilton's equations of motion becomes trivial. Another larger class of transformations called canonoid transformations, which include the canonical ones as a special case, have been studied in relation with the existence of constants of motion (see \cite{Ne,Ca,G.F4,TCA}). It is well known that if a time-independent Hamiltonian system  with $n$ degrees of freedom (the phase space is of dimension $2n$) has $n$ algebraically independent well-defined constants of motion in involution then the system is integrable by quadratures (see Liouville theorem in \cite{Arnold,Babelon,MC}). There exists an analogous statement for time-dependent Hamiltonian systems (see \cite{Gia,LTorres}), namely that for a time-dependent Hamiltonian system with $n$ degrees of freedom (the phase space is now $(2n+1)-$dimensional) the knowledge of $n$ independent integrals of motion in involution allows us to solve the Hamilton's equations by quadratures. There is also a notion of integrability for contact Hamiltonian systems (see \cite{Boyer,Visi}) where constants of motion are involved as well, this is detailed in section \ref{sec7}. It is worth mentioning that even in the simplest symplectic case with $n=2$ the classification of classical and quantum superintegrable Hamiltonian systems (see \cite{AMGen, AMJVPW2015, EWY, AMJVPWIY2018} and references therein) is still an open question.

The goal of the present paper is twofold. Firstly, for Hamiltonian systems we aim to present (under the formalisms of symplectic, cosymplectic, contact and cocontact geometry) the canonical and canonoid transformations as global diffeomorphisms between the phase space to itself intimately connected to the underlying geometric structures. Secondly, it is shown that the existence of a canonoid transformation allows us to find constants of motion which, as has been pointed out before, is fundamental for integrability. It is worth mentioning some differences between our approach and those presented in previous works, for example in \cite{Ne} only the time-independent case is considered whereas in \cite{Ca,G.F4,TCA} the time-dependent case is analyzed but the geometric approach is based on the product structure $\mathbb{R}\times M$ with $M$ a symplectic manifold and the constructions are local only. To the best of the knowledge of the present authors, in this work the notion of canonoid transformations for contact and cocontact Hamiltonian systems is introduced for the first time. We also incorporate to our description in the symplectic and cosymplectic frameworks (Poisson structures) the integrability criterion based on the existence of a Nijenhuis tensor on the phase space. More specifically, given a canonoid transformation for a Hamiltonian system we define a mixed tensor field on the phase space whose traces of its powers are constants of motion. Additionally, for the symplectic and cosymplectic cases, if such a tensor field is a Nijenhuis tensor then the associated constants are pairwise in involution. 

\clearpage

It is worth remarking that the type of Hamiltonian system we deal with is determined by the geometric structure of the corresponding phase space as displayed in the following Table:

\begin{center}
\begin{tabular}{|c|c|}
\hline 
Symplectic structures & Time-independent Hamiltonian systems \\ 
\hline 
Cosymplectic structures & Time-dependent Hamiltonian systems \\ 
\hline 
Contact structures & Contact Hamiltonian systems (time-independent dissipative systems)\\ 
\hline 
Cocontact structures & Cocontact Hamiltonian systems (time-dependent dissipative systems)\\ 
\hline 
\end{tabular} 
\end{center}

This paper is organized as follows: in section \ref{sec2} we briefly review the most important aspects (for our purposes) of the formalism of time-independent Hamiltonian systems on symplectic manifolds. In section \ref{sec3} we present from a geometric perspective the notions of canonical and canonoid transformations for time-independent Hamiltonian systems; here the notion of bi-Hamiltonian systems naturally appears as a closely related concept to our approach, so we introduce some objects and techniques found in this subject as well. A review of cosymplectic geometry and the corresponding formulation of time-dependent Hamiltonian systems is given in section \ref{sec4}. In section \ref{sec5}, the notions of canonical and canonoid transformations for time-dependent Hamiltonian systems are introduced within the framework of cosymplectic geometry; we follow the ideas developed in section \ref{sec3} and despite that the geometric structure differ, analogous results will be obtained. We continue with section \ref{sec6} where we survey the theory of contact geometry and contact Hamiltonian systems; in section \ref{sec7} we introduce the concepts of canonical and canonoid transformations for contact Hamiltonian systems. The last theoretical review is given in  section \ref{sec8}, where we describe the formalism of cocontact Hamiltonian systems. Finally in section \ref{sec9} we introduce the notion of canonoid transformations for cocontact Hamiltonian systems and derive analogous results to those announced for the contact case.

\section{Symplectic geometry and time-independent Hamiltonian systems}
\label{sec2}

In this section, in order to set up the notation and language used in this paper, we present a brief review of symplectic geometry and the formulation of time-independent Hamiltonian mechanics (for details see \cite{AM,LR,G.F2,lee}).
 
Let $(M,\omega)$ be a symplectic manifold of dimension $2n$. Around any point $p\in M$ there exist local coordinates $(q^{1},\cdots,q^{n},p_{1},\cdots,p_{n})$, called canonical coordinates or Darboux coordinates, such that
\begin{equation}
\omega=dq^{i}\wedge dp_{i}.
\end{equation}
In this paper we adopt the Einstein summation convention ( i.e., a summation over repeated indices is assumed).

For each $f\in C^{\infty}(M)$ is assigned a vector field $X_{f}$ on $M$, called the Hamiltonian vector field for $f$, according to
\begin{equation}
X_{f}\lrcorner \omega \ = \ df \ .
\end{equation}
In canonical coordinates, $X_{f}$ reads
\begin{equation}
X_{f}=\frac{\partial f}{\partial p_{i}}\frac{\partial}{\partial q^{i}}-\frac{\partial f}{\partial q^{i}}\frac{\partial}{\partial p_{i}}\ .
\end{equation}
The assignment $f\longmapsto X_{f}$ is linear, that is
\begin{equation}
X_{f+\alpha g}\ = \ X_{f}+\alpha X_{g}\ ,
\end{equation}
$\forall f,g\in C^{\infty}(M)$ and $\forall \alpha \in\mathbb{R}$. Given $f,g \in C^{\infty}(M)$ the Poisson bracket of $f$ and $g$ is defined by
\begin{equation}
\lbrace f,g\rbrace=X_{g}(f)=\omega(X_{f},X_{g}).
\end{equation}
In canonical coordinates we have
\begin{equation}
\lbrace f,g\rbrace \ = \ \frac{\partial f}{\partial q^{i}}\frac{\partial g}{\partial p_{i}}\,-\,\frac{\partial f}{\partial p_{i}}\frac{\partial g}{\partial q^{i}} \ .
\end{equation}

The theory of time-independent Hamiltonian systems is naturally constructed within the mathematical formalism of symplectic geometry. Given $H\in C^{\infty}(M)$ the dynamics of the Hamiltonian system on $(M,\omega)$ (the phase space) with Hamiltonian function $H$ is defined by the Hamiltonian vector field $X_{H}$, that is, the trajectories of the system $\varphi(t)=(q^{1}(t),\cdots,q^{n}(t),p_{1}(t),\cdots,p_{n}(t))$ are the integral curves of $X_{H}$, they satisfy the Hamilton's equations of motion
\begin{equation}
\dot{q^{i}} =\frac{\partial H}{\partial p_{i}}, \hspace{1cm}
\dot{p_{i}} =-\frac{\partial H}{\partial q^{i}}\qquad ;\qquad i=1,2,3,\ldots,n \ .
\end{equation}

The evolution (the temporal evolution) of a function $f\in C^{\infty}(M)$ (a physical observable) along the trajectories of the system is given by
\begin{equation}
\dot{f}= \mathcal{L}_{X_{H}}f=X_{H}f=\lbrace f,H\rbrace \ ,
\end{equation}
where $L_{X_{H}}f$ is the Lie derivative of $f$ with respect to $X_{H}$. We say that $f$ is a constant of motion of the system if it is constant along the trajectories of the system, that is, $f$ is a constant of motion if $L_{X_{H}}f=0$ ($\lbrace f,H\rbrace=0$). As mentioned in the Introduction, the existence of a constant of motion effectively reduces the dimension of the phase space and the integration of the equations of motion becomes simpler.

\section{Canonical and canonoid transformations for time-independent Hamiltonian systems}
\label{sec3}

The concept of canonical transformations for time-independent Hamiltonian systems is well known, on the other hand, the concept of canonoid transformation is not so popular. In this section we present the latter concept from a geometric perspective. This section establishes a guide to introduce canonoid transformations in cosymplectic as well as in contact Hamiltonian mechanics. 

Let us consider a symplectic manifold $(M,\omega)$ of dimension $2n$.
\begin{de}
We say that a diffeomorphism $F:M\longrightarrow M$ is a canonical transformation if $F^{*}\omega=\omega$.
\end{de} 
Let $F:M\longrightarrow M$ be a diffeomorphism. We know that $F^{*}\omega$ is a symplectic structure on $M$, let us denote it by $\overline{\omega}$ and by $\overline{\lbrace,\rbrace}$ the Poisson bracket defined by it. Around any point $p\in M$ there are local coordinates $(Q^{1},\cdots,Q^{n},P_{1},\cdots,P_{n})$ such that
\begin{equation}
\overline{\omega}=dQ^{i}\wedge dP_{i}.
\end{equation} 
Given $f\in C^{\infty}(M)$, let us denote by $\overline{X}_{f}$ the Hamiltonian vector field for $f$ defined by $\overline{\omega}$, then $\overline{\lbrace f,g\rbrace}=\overline{X}_{g}f$. The following propositions are equivalent
\begin{enumerate}
\item $F$ is a canonical transformation,
\item $\overline{X}_{f}=X_{f}$ $\forall f \in C^{\infty}(M)$ and
\item $\overline{\lbrace f,g\rbrace}=\lbrace f,g\rbrace$ $\forall f,g \in C^{\infty}(M)$.
\end{enumerate}
So, a diffeomorphism from a symplectic manifold to itself is a canonical transformation if and only if it preserves the Poisson bracket. 

If we think of $F$ locally as a coordinate transformation $(q^{1},\cdots,q^{n},p_{1},\cdots,p_{n})\longmapsto (Q^{1},\cdots,Q^{n},P_{1},\cdots,P_{n})$  then it is a canonical transformation if and only if
\begin{equation}\label{eqcpb}
\lbrace Q^{i},Q^{j}\rbrace=\lbrace P_{i},P_{j}\rbrace=0 \hspace{1cm}\textit{and}\hspace{1cm}\lbrace Q^{i},P_{j}\rbrace=\delta^{i}_{j}.
\end{equation}
Classically (see \cite{MC}), canonical transformations have been defined as changes of canonical coordinates, i.e, coordinate transformations from canonical coordinates such that the new coordinates satisfy equations (\ref{eqcpb}) (canonical Poisson brackets).

Observe that up to this point in this section we have not involved any Hamiltonian system. Let us consider the Hamiltonian system $(M,\omega,H)$; if $F$ is a canonical transformation then the equations of motion in the coordinates $(Q^{1},\cdots,Q^{n},P_{1},\cdots,P_{n})$ are
\begin{equation}
\begin{split}
\dot{Q^{i}} &= \lbrace Q^{i},H\rbrace=\overline{\lbrace Q^{i},H\rbrace}=\frac{\partial H}{\partial P_{i}},\\
\dot{P_{i}} &=\lbrace P_{i},H\rbrace=\overline{\lbrace P_{i},H\rbrace}=-\frac{\partial H}{\partial Q^{i}}\ ,
\end{split}
\end{equation}
hence, the canonical transformations preserve the Hamiltonian form of the equations of motion. There are transformations for Hamiltonian systems that preserve the Hamiltonian form of equations of motion and they are not necessarily canonical transformations, we called canonoid transformations these more general transformations.

\begin{de}
We say that the diffeomorphism $F:M\longrightarrow M$ is a canonoid transformation for the Hamiltonian system $(M,\omega,H)$ if there exists a function $K\in C^{\infty}(M)$ such that 
\begin{equation}
X_{H}\lrcorner \overline{\omega}=dK.
\end{equation}
\end{de}
Observe that every canonical transformation is also a canonoid transformation.

For the rest of this section let us suppose that $F$ is a canonoid transformation. Given $f\in C^{\infty}(M)$ we have that
\begin{equation}
\dot{f}=\lbrace f,H\rbrace=\overline{\lbrace f,K\rbrace}.
\end{equation}
In general, canonoid transformations do not preserve the Poisson bracket, so that if we think of $F$ locally as a coordinate transformation $(q^{1},\cdots,q^{n},p_{1},\cdots,p_{n})\longmapsto (Q^{1},\cdots,Q^{n},P_{1},\cdots,P_{n})$  then the new coordinates do not necessarily satisfy equations (\ref{eqcpb}). The equations of motion in the coordinates $(Q^{1},\cdots,Q^{n},P_{1},\cdots,P_{n})$ read
\begin{equation}
\begin{split}
\dot{Q^{i}} &= \lbrace Q^{i},H\rbrace=\overline{\lbrace Q^{i},K\rbrace}=\frac{\partial K}{\partial P_{i}},\\
\dot{P_{i}} &=\lbrace P_{i},H\rbrace=\overline{\lbrace P_{i},K\rbrace}=-\frac{\partial K}{\partial Q^{i}}.
\end{split}
\end{equation}
Therefore, the canonoid transformations preserve the Hamiltonian form of the equations of motion. We can observe that reciprocally, transformations that preserve the Hamiltonian form of the equations of motion are canonoid. Canonoid transformations have been defined as coordinate transformations preserving the Hamiltonian form of the equations of motion (see \cite{G.F4}). There exists a relationship between $H$ and $K$ (see \cite{G.F4}), a straightforward computation gives
\begin{equation}
\frac{\partial K}{\partial p_{l}} = [p_{l}, p_{j}] \frac{\partial H}{\partial q^{j}} - [p_{l}, q^{j}] \frac{\partial H}{\partial p_{j}}\qquad {\rm and} \qquad \frac{\partial K}{\partial q^{l}} = [q^{l}, p_{j}] \frac{\partial H}{\partial q^{j}} - [q^{l}, q^{j}] \frac{\partial H}{\partial p_{j}}.
\end{equation}

In this case we have two Hamiltonian systems $(M,\omega,H)$ and $(M,\overline{\omega},K)$ with the same phase space and the dynamics defined by the same Hamiltonian vector field $X_{H}$ (or two Hamiltonian representations for the same mechanical system). The symplectic structures $\omega$ and $\overline{\omega}$ are invariant under the flow of $X_{H}$, indeed
\begin{equation}
L_{X_{H}}\omega = X_{H}\lrcorner d\omega+d(X_{H}\lrcorner \omega)=d(X_{H}\lrcorner \omega)=d(dH)=0
\end{equation}
and analogously we see that $L_{X_{H}}\overline{\omega}=0$. Thus, we arrive to one condition for a bi-Hamiltonian system. The concept of bi-Hamiltonian system was introduced by F. Magry and C. Morosi \cite{MM} as follows: a (time-independent) Hamiltonian system with symplectic structure $\gamma$ is called bi-Hamiltonian if there exists a second symplectic structure $\overline{\gamma}$ invariant under the flow of the corresponding Hamiltonian vector field, such that the (1,1)-tensor field $S$ defined by $\overline{\gamma}=S\lrcorner \gamma$ is a Nijenhuis tensor ($(S\lrcorner\gamma)(X,Y)=\gamma(SX,Y)$). A Nijenhuis tensor is a (1,1)-tensor field whose Nijenhuis torsion tensor vanishes; the Nijenhuis torsion tensor of a (1,1)-tensor field $L$ is defined by
\begin{equation}
T(L)(X,Y) = [LX, LY] - L[LX,Y] - L[X, LY] + L^{2}[X,Y] \ ,
\end{equation}
for any vector fields $X,Y$. A symplectic manifold equipped with a Nijenhuis tensor is called a symplectic-Nijenhuis manifold (or a $\omega N$ manifold) (for details see \cite{MM,KM,FP}). Given a bi-Hamiltonian system with Nijenhuis tensor $S$, the traces of the powers of $S$ (also the eigenvalues) are constants of motion in involution (see \cite{MM,RB}).  

Following the work of F. Magry and C. Morosi on bi-Hamiltonian systems, we can define an (1,1)-tensor field $S$ on $M$ by
\begin{equation}\label{eqS}
\overline{\omega}=S\lrcorner \omega.
\end{equation}
This relation (\ref{eqS}) allows us to obtain 
an analogous result to that aforementioned for bi-Hamiltonian systems.

\begin{te}\label{te1}
If $F:M\longrightarrow M$ is a canonoid transformation for the Hamiltonian system $(M,\omega,H)$, then the traces of the powers of the (1,1)-tensor field $S$ defined by equation (\ref{eqS}) are constants of motion.
\end{te}

\begin{proof}
\begin{equation}
\begin{split}
L_{X_{H}}\overline{\omega}=(L_{X_{H}}S)\lrcorner \omega+S\lrcorner(L_{X_{H}}\omega)&\Longrightarrow  0=(L_{X_{H}}S)\lrcorner \omega\\
&\Longrightarrow L_{X_{H}}S=0. 
\end{split}
\end{equation}
Then for $l=1,2,\cdots$, we have
\begin{equation}
L_{X_{H}}tr(S^{l})=tr(L_{X_{H}}S^{l})=\sum_{i=1}^{l-1}tr(S^{i-1}(L_{X_{H}}S)S^{l-i})=0.
\end{equation}
\end{proof}

It is worth mentioning that the maximum number of functionally independent constants of motion obtained in this way is $n$. This result follows from the fact that the characteristic polynomial of the product of two antisymmetric $2n \times 2n$ matrices is the square of a polynomial of degree $n$ (see \cite{G.F4}). 

So far, we have not considered the Nijenhuis torsion tensor of $S$, in fact if it vanishes then $(M,\omega)$ becomes a symplectic-Nijenhuis manifold and the traces of the powers of $S$ are constants of motion in involution. Indeed, let us write down $S$ in canonical coordinates $(q^{1},\cdots,q^{n},p_{1},\cdots,p_{n})$. In order to have a compact local form of the relevant geometric objects, we write $(x^{1},\cdots,x^{2n})=(q^{1},\cdots,q^{n},p_{1},\cdots,p_{n})$ then we have
\begin{equation}
\omega=\frac{1}{2}\epsilon_{\mu\nu}dx^{\mu}\wedge dx^{\nu},
\end{equation}
where $\epsilon_{\mu\nu}$ is the entry in the $\mu$-th row and the $\nu$-th column of the $2n\times 2n$  matrix $\epsilon=\left(\begin{array}{cc} 0 & I \\ -I & 0 \end{array}\right)$. Also, 
\begin{equation}
X_{H}=\epsilon^{\mu\nu}\frac{\partial H}{\partial x^{\mu}}\frac{\partial}{\partial x^{\nu}},
\end{equation}
here $\epsilon^{\mu\nu}$ denote the entries of the inverse matrix of $\epsilon$, namely $\epsilon^{-1}=\left(\begin{array}{cc} 0 & -I \\ I & 0 \end{array}\right)$ ($\epsilon^{\mu\nu}$ are the local components of the Poisson tensor defined by the symplectic form $\omega$). The matrices $\epsilon$ and $\epsilon^{-1}$ are antisymmetric and $\epsilon_{\mu\nu}=\epsilon^{\nu\mu}$. The corresponding equations of motion are
\begin{equation}
\dot{x}^{\mu}=\epsilon^{\mu\nu}\frac{\partial H}{\partial x^{\nu}}.
\end{equation}
On the other hand
\begin{equation}
\overline{\omega}=dQ^{i}\wedge dP_{i}=\frac{\partial Q^{i}}{\partial x^{\mu}}\frac{\partial P_{i}}{\partial x^{\nu}}dx^{\mu}\wedge dx^{\nu}=\frac{1}{2}[x^{\mu},x^{\nu}]dx^{\mu}\wedge dx^{\nu},
\end{equation}
where $[x^{\mu},x^{\nu}]=\frac{\partial Q^{i}}{\partial x^{\mu}}\frac{\partial P_{i}}{\partial x^{\nu}}-\frac{\partial Q^{i}}{\partial x^{\nu}}\frac{\partial P_{i}}{\partial x^{\mu}}$ is the Lagrange bracket of $x^{\mu}$ and $x^{\nu}$. Now
\begin{equation*}
\frac{1}{2}[x^{\mu},x^{\nu}]dx^{\mu}\wedge dx^{\nu}=(S^{\alpha}_{\beta}\frac{\partial}{\partial x^{\alpha}}\otimes dx^{\beta})\lrcorner \frac{1}{2}\epsilon_{\mu\nu}dx^{\mu}\wedge dx^{\nu}
\end{equation*}
which implies that 
\begin{equation*}
[x^{\beta},x^{\nu}]=\epsilon_{\alpha\nu}S^{\alpha}_{\beta} \ ,
\end{equation*}
as well as
\begin{equation*}
S^{\alpha}_{\beta}=\epsilon^{\lambda\alpha}[x^{\beta},x^{\lambda}];
\end{equation*}
so we have that the expression of $S$ in the canonical coordinates $(x^{1},\cdots,x^{2n})=(q^{1},\cdots,q^{n},p_{1},\cdots,p_{n})$ takes the form
\begin{equation}
S=\epsilon^{\lambda\alpha}[x^{\beta},x^{\lambda}]\frac{\partial}{\partial x^{\alpha}}\otimes dx^{\beta}.
\end{equation}
It is known (see \cite{MM,KM,Das}) that if we have two Poisson brackets on a manifold and a sequence of functions satisfying the Lenard recursion relations for such Poisson brackets then the functions in that sequence are pairwise in involution with respect to both Poisson brackets. For two Poisson brackets $\lbrace\rbrace_{1}$ and $\lbrace\rbrace_{2}$ (or two Poisson tensors $P_{1}$ and $P_{2}$), and a sequence of functions $(f_{k})_{k=1}^{\infty}$, the Lenard recursion relations read
\begin{equation}
\lbrace f_{k},\cdot\rbrace_{1}=\lbrace f_{k+1},\cdot\rbrace_{2}\hspace{1cm}(P_{1}(df_{k})=P_{2}(df_{k+1})).
\end{equation}
In our case we have two Poisson brackets $\lbrace,\rbrace$ and $\overline{\lbrace,\rbrace}$ (or two Poisson tensors $P$ and $\overline{P}$ defined by $\omega$ and $\overline{\omega}$ respectively, and the Nijenhuis tensor $S$ is defined as in \cite{KM} by $S=P\overline{P}^{-1}$). For the sequence of functions $(tr(S^{k}))_{k=1}^{\infty}$ , the Lenard recursion relations are
\begin{equation}\label{LRR}
\lbrace tr(S^{k}),\cdot\rbrace=\overline{\lbrace tr(S^{k+1}),\cdot\rbrace}.
\end{equation}
In canonical coordinates $(x^{1},\cdots,x^{2n})=(q^{1},\cdots,q^{n},p_{1},\cdots,p_{n})$ the Lenard recursion relations (\ref{LRR}) have the form
\begin{equation}
\epsilon^{\lambda\alpha}\frac{\partial tr(S^{k})}{\partial x^{\alpha}}=\overline{\epsilon}^{\lambda\beta}\frac{\partial tr(S^{k+1})}{\partial x^{\beta}},
\end{equation}
where $\overline{\epsilon}^{\lambda\beta}$ are the entries of the inverse matrix of $(\overline{\epsilon}_{\lambda\beta})=([x^{\lambda},x^{\beta
}])$ ($\overline{\epsilon}^{\lambda\beta}$ are the local components of the Poisson tensor defined by the symplectic form $\overline{\omega}$). Thus,
\begin{equation}
\begin{split}
\epsilon^{\lambda\alpha}\frac{\partial tr(S^{k})}{\partial x^{\alpha}}=\overline{\epsilon}^{\lambda\beta}\frac{\partial tr(S^{k+1})}{\partial x^{\beta}} &\Leftrightarrow \epsilon^{\lambda\alpha}\overline{\epsilon}_{\beta\lambda}\frac{\partial tr(S^{k})}{\partial x^{\alpha}}=\frac{\partial tr(S^{k+1})}{\partial x^{\beta}}\\
&\Leftrightarrow S^{\alpha}_{\beta}\frac{\partial tr(S^{k})}{\partial x^{\alpha}}=\frac{\partial tr(S^{k+1})}{\partial x^{\beta}}\\
&\Leftrightarrow S^{\alpha}_{\beta}\frac{\partial tr(S^{k})}{\partial x^{\alpha}}-\frac{\partial tr(S^{k+1})}{\partial x^{\beta}}=0.
\end{split}
\end{equation}
The component $N^{\lambda}_{\beta\gamma}$ of the Nijenhuis torsion tensor of $S$ is
\begin{equation}
\frac{\partial S_{\gamma}^{\lambda}}{\partial x^{\nu}}S_{\beta}^{\nu}-\frac{\partial S_{\beta}^{\lambda}}{\partial x^{\nu}}S_{\gamma}^{\nu}+\left( \frac{\partial S_{\beta}^{\nu}}{\partial x^{\gamma}}-\frac{\partial S_{\gamma}^{\nu}}{\partial x^{\beta}}\right)S_{\nu}^{\lambda}.
\end{equation}
We have that
\begin{equation}
N_{\beta\gamma}^{\lambda}(S^{k-1})_{\lambda}^{\gamma}=S^{\alpha}_{\beta}\frac{\partial tr(S^{k})}{\partial x^{\alpha}}-\frac{\partial tr(S^{k+1})}{\partial x^{\beta}}.
\end{equation}
Therefore, if the Nijenhuis torsion tensor of $S$ vanishes then the sequence $(tr(S^{k}))_{k=1}^{\infty}$ satisfies the Lenard recursion relations, that is, if $S$ is a Nijenhuis tensor then the traces of the powers of $S$ are functions pairwise in involution.

\section{Cosymplectic geometry and time-dependent Hamiltonian systems}
\label{sec4}

In this section, a review of the cosymplectic geometry and the associated formulation of time-dependent Hamiltonian systems under cosymplectic manifolds is presented. As in the previous section \ref{sec2}, the purpose of this is to set up notation and language. 

\begin{de}
Let $M$ be a $2n+1$ dimensional smooth manifold. A cosymplectic structure on $M$ is a couple $(\Omega,\eta)$, where $\Omega$ is a closed 2-form on $M$ and $\eta$ is a closed 1-form on $M$ such that $\eta\wedge\Omega^{n}\neq 0$. If $(\Omega,\eta)$ is a cosymplectic structure on $M$ we say that $(M,\Omega,\eta)$ is a cosymplectic manifold.
\end{de}

Let $(M,\Omega,\eta)$ be a cosymplectic manifold of dimension $2n+1$. Around any point $p\in M$ there exist local coordinates $(q^{1},\cdots,q^{n},p_{1},\cdots,p_{n},t)$, called canonical coordinates or Darboux coordinates, such that
\begin{equation}
\Omega=dq^{i}\wedge dp_{i}\hspace{1cm}\textit{and}\hspace{1cm}\eta=dt.
\end{equation}
There exists a distinguished vector field $R$ on $M$, called the Reeb vector field, which obeys
\begin{equation}
R\lrcorner \Omega =0 \hspace{1cm}\textit{and}\hspace{1cm} R\lrcorner \eta =1.
\end{equation}
In canonical coordinates we have $R=\frac{\partial}{\partial t}$.

For each $f\in C^{\infty}(M)$ is assigned a vector field $X_{f}$ on $M$, called the Hamiltonian vector field for $f$, according to
\begin{equation}
X_{f}\lrcorner \Omega =df-R(f)\eta \hspace{1cm}\textit{and}\hspace{1cm} X_{f}\lrcorner \eta =0.
\end{equation}
In canonical coordinates we have
\begin{equation}
X_{f}=\frac{\partial f}{\partial p_{i}}\frac{\partial}{\partial q^{i}}-\frac{\partial f}{\partial q^{i}}\frac{\partial}{\partial p_{i}}.
\end{equation}
The assignment $f\longmapsto X_{f}$ is linear, that is
\begin{equation}
X_{f+\alpha g}=X_{f}+\alpha X_{g},
\end{equation}
$\forall f,g\in C^{\infty}(M)$ and $\forall \alpha \in\mathbb{R}$. Given $f,g \in C^{\infty}(M)$ the Poisson bracket of $f$ and $g$ is defined by
\begin{equation}
\lbrace f,g\rbrace=X_{g}(f)=\Omega(X_{f},X_{g}).
\end{equation}
In canonical coordinates, it reads
\begin{equation}
\lbrace f,g\rbrace=\frac{\partial f}{\partial q^{i}}\frac{\partial g}{\partial p_{i}}-\frac{\partial f}{\partial p_{i}}\frac{\partial g}{\partial q^{i}}.
\end{equation}

The theory of time-dependent Hamiltonian systems can be developed under the mathematical formalism of cosymplectic geometry (see \cite{LR,Cantr,LS}). Given $H\in C^{\infty}(M)$ the dynamics of the Hamiltonian system on $(M,\Omega,\eta)$ (the phase space) with Hamiltonian function $H$ is defined by the evolution vector field $E_{H}=X_{H}+R$. $E_{H}$ is the only vector field on $M$ such that 
\begin{equation}
E_{H}\lrcorner (\Omega+dH\wedge\eta)=0 \hspace{1cm}\textit{and}\hspace{1cm} E_{H}\lrcorner \eta =1.
\end{equation}
In canonical coordinates we have
\begin{equation}
E_{H}=\frac{\partial H}{\partial p_{i}}\frac{\partial}{\partial q^{i}}-\frac{\partial H}{\partial q^{i}}\frac{\partial}{\partial p_{i}}+\frac{\partial}{\partial t}.
\end{equation}
The trajectories $\psi(s)=(q^{1}(s),\cdots,q^{n}(s),p_{1}(s),\cdots,p_{n}(s),t(s))$ of the system are the integral curves of $E_{H}$, they satisfy the Hamilton equations of motion
\begin{equation}
\dot{q^{i}} =\frac{\partial H}{\partial p_{i}}, \hspace{1cm}
\dot{p_{i}} =-\frac{\partial H}{\partial q^{i}}, \hspace{1cm}
\dot{t}=1.
\end{equation}
The condition $\dot{t}=1$ (then, $t=s$) implies that the temporal parameter for the system is $t$, that is, the trajectories of the system are parametrized by $t$
\begin{equation}
\psi(t)=(q^{1}(t),\cdots,q^{n}(t),p_{1}(t),\cdots,p_{n}(t),t).
\end{equation}

The evolution of a function $f\in C^{\infty}(M)$ (an observable) along the trajectories of the system is given by
\begin{equation}
\dot{f}= \mathcal{L}_{E_{H}}f=E_{H}f=X_{H}f+R(f)=\lbrace f,H\rbrace+\frac{\partial f}{\partial t}.
\end{equation}
We say that a function $f\in C^{\infty}(M)$ as a constant of motion of the system if it is constant along the trajectories of the system, that is, $f$ is a constant of motion if $L_{E_{H}}f=0$ ($\lbrace f,H\rbrace+\frac{\partial f}{\partial t}=0$).

\section{Canonical and canonoid transformations for time-dependent Hamiltonian systems}
\label{sec5}

In this section we introduce canonoid transformations for time-dependent Hamiltonian systems under the formalism of cosymplectic geometry. The main result is presented in theorem \ref{te2}, it is analogous to theorem \ref{te1} for time-independent Hamiltonian system with symplectic geometry. It is worth mentioning that, as it was stated in the introduction, there are some differences between our approach and those presented in previous works such as \cite{Ne,Ca,G.F4,TCA}.

Now, let us consider a cosymplectic manifold $(M,\Omega,\eta)$ of dimension $2n+1$.
\begin{de}
We say that a diffeomorphism $F:M\longrightarrow M$ is a canonical transformation if $F^{*}\Omega=\Omega$ and $F^{*}\eta=\eta$.
\end{de} 
Let $F:M\longrightarrow M$ be a diffeomorphism. We know that $(F^{*}\Omega,F^{*}\eta)$ is a cosymplectic structure on $M$, let us denote it by $(\overline{\Omega},\overline{\eta})$, and by $\overline{\lbrace,\rbrace}$ the Poisson bracket defined by it and by $\overline{R}$ the Reeb vector field for this cosymplectic structure. Around any point $p\in M$ there are local coordinates $(Q^{1},\cdots,Q^{n},P_{1},\cdots,P_{n},T)$ such that
\begin{equation}
\overline{\Omega}=dQ^{i}\wedge dP_{i}\hspace{1cm},\hspace{1cm}\overline{\eta}=dT\hspace{1cm}\textit{and}\hspace{1cm}\overline{R}=\frac{\partial}{\partial T}.
\end{equation} 
Given $f\in C^{\infty}(M)$, let us denote by $\overline{X}_{f}$ the Hamiltonian vector field for $f$ defined by $(\overline{\Omega},\overline{\eta})$, then $\overline{\lbrace f,g\rbrace}=\overline{X}_{g}f$. The following propositions are equivalent
\begin{enumerate}
\item $F$ is a canonical transformation,
\item $\overline{X}_{f}=X_{f}$ $\forall f \in C^{\infty}(M)$ and
\item $\overline{\lbrace f,g\rbrace}=\lbrace f,g\rbrace$ $\forall f,g \in C^{\infty}(M)$.
\end{enumerate}
So, a diffeomorphism from a cosymplectic manifold to itself is a canonical transformation if a and only if it preserves the Poisson bracket. 

If we think of $F$ locally as a coordinate transformation $(q^{1},\cdots,q^{n},p_{1},\cdots,p_{n},t)\longmapsto (Q^{1},\cdots,Q^{n},P_{1},\cdots,P_{n},T)$  then it is a canonical transformation if and only if
\begin{equation}
\lbrace Q^{i},Q^{j}\rbrace=\lbrace P_{i},P_{j}\rbrace=0,\hspace{1cm}\lbrace Q^{i},P_{j}\rbrace=\delta^{i}_{j}\hspace{1cm}\textit{and}\hspace{1cm}T=t.
\end{equation}

Let us consider the Hamiltonian system $(M,\Omega,\eta,H)$; if $F$ is a canonical transformation then the equations of motion in the coordinates $(Q^{1},\cdots,Q^{n},P_{1},\cdots,P_{n},t)$ are
\begin{equation}
\begin{split}
\dot{Q^{i}} &= \lbrace Q^{i},H\rbrace=\overline{\lbrace Q^{i},H\rbrace}=\frac{\partial H}{\partial P_{i}},\\
\dot{P_{i}} &=\lbrace P_{i},H\rbrace=\overline{\lbrace P_{i},H\rbrace}=-\frac{\partial H}{\partial Q^{i}}.
\end{split}
\end{equation}
Hence, canonical transformations preserve the Hamiltonian form of the equations of motion. As in the symplectic case, there are transformations for Hamiltonian systems that preserve the Hamiltonian form of equations of motion and they are not necessarily canonical transformations, we also called canonoid transformations this more general transformations for Hamiltonian systems.

\begin{de}
We say that the diffeomorphism $F:M\longrightarrow M$ is a canonoid transformation for the Hamiltonian system $(M,\Omega,\eta,H)$ if there exists a function $K\in C^{\infty}(M)$ such that 
\begin{equation}
X_{H}\lrcorner \overline{\Omega}=dK-\overline{R}(K)\overline{\eta}\hspace{1cm}\textit{and}\hspace{1cm}X_{H}\lrcorner \overline{\eta}=0.
\end{equation}
\end{de}
Again we can observe that every canonical transformation is also a canonoid transformation.

For the rest of this section let us suppose that $F$ is a canonoid transformation. We have two Hamiltonian systems $(M,\Omega,\eta,H)$ and $(M,\overline{\Omega},\overline{\eta},K)$; the temporal parameter for the first one is $t$ whilst the temporal parameter for the second one is $T$, then $T=t$ (formally $T=F^{*}t$), so we have that $\overline{\eta}=\eta$ and $\overline{R}=R$. The dynamics of both systems is defined by the same evolution vector field $E_{H}$ (i.e., we have two Hamiltonian representations for the same mechanical system).

Given $f\in C^{\infty}(M)$ we have that
\begin{equation}
\dot{f}=\lbrace f,H\rbrace+Rf=\overline{\lbrace f,K\rbrace}+Rf,
\end{equation}
then $\lbrace f,H\rbrace=\overline{\lbrace f,K\rbrace}$.
In general, canonoid transformations do not preserve the Poisson bracket. The equations of motion in the coordinates $(Q^{1},\cdots,Q^{n},P_{1},\cdots,P_{n},t)$ are
\begin{equation}
\begin{split}
\dot{Q^{i}} &= \lbrace Q^{i},H\rbrace+\frac{\partial Q^{i}}{\partial t}=\overline{\lbrace Q^{i},K\rbrace}+\frac{\partial Q^{i}}{\partial T}=\frac{\partial K}{\partial P_{i}},\\
\dot{P_{i}} &=\lbrace P_{i},H\rbrace+\frac{\partial P_{i}}{\partial t}=\overline{\lbrace P_{i},K\rbrace}+\frac{\partial P_{i}}{\partial T}=-\frac{\partial K}{\partial Q^{i}}\ ,
\end{split}
\end{equation}
so that the Hamiltonian form of the equations of motion is invariant under canonoid transformations. We can see that, reciprocally, transformations that preserve the Hamiltonian form of the equations of motion are canonoid. 

In this case, the 2-forms $\Omega$ and $\overline{\Omega}$ are not invariant under the flow of $E_{H}$, indeed
\begin{equation}
L_{E_{H}}\Omega = E_{H}\lrcorner d\Omega+d(E_{H}\lrcorner \Omega)=d(E_{H}\lrcorner \Omega)=d((X_{H}+R)\lrcorner\Omega)=d(X_{H}\lrcorner\Omega)=d(dH-RH\eta)=-d(RH)\wedge\eta \ ,
\end{equation}
and analogously we see that 
\begin{equation}
L_{E_{H}}\overline{\Omega}=-d(RK)\wedge\eta.
\end{equation}
Despite this, similarly to the symplectic case we define the (1,1)-tensor field $S$ on $M$ by
\begin{equation}
\overline{\Omega}=S\lrcorner \Omega \hspace{1cm}\textit{and}\hspace{1cm}S\lrcorner\eta=0.
\label{STDH} 
\end{equation}
In this case $S$ is not invariant under the flow of $E_{H}$, but as in the symplectic case, the traces of the powers of $S$ are constants of motion.

\begin{te}\label{te2}
If $F:M\longrightarrow M$ is a canonoid transformation for the time-dependent Hamiltonian system $(M,\Omega,\eta,H)$, then the traces of the powers of the (1,1)-tensor field $S$ defined in (\ref{STDH}) are constants of motion.
\end{te}

\begin{proof}
First let us see $S$ in canonical coordinates $(q^{1},\cdots,q^{n},p_{1},\cdots,p_{n},t)$. As in section \ref{sec3}, in order to have a compact local form of the geometric objects, we write $(x^{1},\cdots,x^{2n})=(q^{1},\cdots,q^{n},p_{1},\cdots,p_{n})$ then we have
\begin{equation}
\Omega=\frac{1}{2}\epsilon_{\mu\nu}dx^{\mu}\wedge dx^{\nu},
\end{equation}
where $\epsilon_{\mu\nu}$ is the entry the $\mu$-th row and the $\nu$-th column of the symplectic matrix $\epsilon$ defined in section \ref{sec3}.
\begin{equation}
E_{H}=\epsilon^{\mu\nu}\frac{\partial H}{\partial x^{\mu}}\frac{\partial}{\partial x^{\nu}}+\frac{\partial}{\partial t},
\end{equation}
here $\epsilon^{\mu\nu}$ is the entry the $\mu$-th row and the $\nu$-th column of the inverse matrix of $\epsilon$. The corresponding equations of motion read
\begin{equation}
\dot{x}^{\mu}=\epsilon^{\mu\nu}\frac{\partial H}{\partial x^{\nu}}.
\end{equation}
The local form of the tensor $S$ is given by
\begin{equation}
S=S^{\alpha}_{\beta}\frac{\partial}{\partial x^{\alpha}}\otimes dx^{\beta}+S^{\alpha}_{t}\frac{\partial}{\partial x^{\alpha}}\otimes dt+S^{t}_{\beta}\frac{\partial}{\partial t}\otimes dx^{\beta}+S^{t}_{t}\frac{\partial}{\partial t}\otimes dt.
\end{equation}
Since $S\lrcorner\eta=0$ then $S^{t}_{\beta}=S^{t}_{t}=0$; so that
\begin{equation}
S=S^{\alpha}_{\beta}\frac{\partial}{\partial x^{\alpha}}\otimes dx^{\beta}+S^{\alpha}_{t}\frac{\partial}{\partial x^{\alpha}}\otimes dt.
\end{equation}
Also, we have that
\begin{equation}
S^{\alpha}_{\beta}=\epsilon^{\lambda\alpha}[x^{\beta},x^{\lambda}]\hspace{1cm}\textit{and}\hspace{1cm}S^{\alpha}_{t}=\epsilon^{\lambda\alpha}[t,x^{\lambda}] \ ,
\end{equation}
indeed; first let us see the local expression of $\overline{\Omega}$
\begin{equation}
\overline{\Omega}=dQ^{i}\wedge dP_{i}=\frac{\partial Q^{i}}{\partial x^{\mu}}\frac{\partial P_{i}}{\partial x^{\nu}}dx^{\mu}\wedge dx^{\nu}+\frac{\partial Q^{i}}{\partial x^{\mu}}\frac{\partial P_{i}}{\partial t}dx^{\mu}\wedge dt+\frac{\partial Q^{i}}{\partial t}\frac{\partial P_{i}}{\partial x^{\nu}}dt\wedge dx^{\nu}=\frac{1}{2}[x^{\mu},x^{\nu}]dx^{\mu}\wedge dx^{\nu}+[x^{\mu},t]dx^{\mu}\wedge dt,
\end{equation}
where $[,]$ denotes the Lagrange bracket as in section \ref{sec3}. Thus, we have that
\begin{equation}
\frac{1}{2}[x^{\mu},x^{\nu}]dx^{\mu}\wedge dx^{\nu}+[x^{\mu},t]dx^{\mu}\wedge dt=(S^{\alpha}_{\beta}\frac{\partial}{\partial x^{\alpha}}\otimes dx^{\beta}+S^{\alpha}_{t}\frac{\partial}{\partial x^{\alpha}}\otimes dt)\lrcorner \frac{1}{2}\epsilon_{\mu\nu}dx^{\mu}\wedge dx^{\nu}
\end{equation}
which implies that
\begin{equation*}
[x^{\beta},x^{\nu}]=\epsilon_{\alpha\nu}S^{\alpha}_{\beta}\hspace{1cm}\textit{and}\hspace{1cm}[t,x^{\nu}]=\epsilon_{\alpha\nu}S^{\alpha}_{t},
\end{equation*}
as well as
\begin{equation*}
S^{\alpha}_{\beta}=\epsilon^{\lambda\alpha}[x^{\beta},x^{\lambda}]\hspace{1cm}\textit{and}\hspace{1cm}S^{\alpha}_{t}=\epsilon^{\lambda\alpha}[t,x^{\lambda}] \ .
\end{equation*}
Consequently, the expression of $S$ in the canonical coordinates $(x^{1},\cdots,x^{2n})=(q^{1},\cdots,q^{n},p_{1},\cdots,p_{n})$ is the following
\begin{equation}
S=\epsilon^{\lambda\alpha}[x^{\beta},x^{\lambda}]\frac{\partial}{\partial x^{\alpha}}\otimes dx^{\beta}+\epsilon^{\lambda\alpha}[t,x^{\lambda}]\frac{\partial}{\partial x^{\alpha}}\otimes dt.
\end{equation}
Now, we observe that in canonical coordinates 
\begin{equation}
\begin{split}
L_{E_{H}}S &=\epsilon^{\lambda\alpha}\left( (L_{E_{H}}[x^{\beta},x^{\lambda}])\frac{\partial}{\partial x^{\alpha}}\otimes dx^{\beta}+[x^{\beta},x^{\lambda}](L_{E_{H}}\frac{\partial}{\partial x^{\alpha}})\otimes dx^{\beta}+[x^{\beta},x^{\lambda}]\frac{\partial}{\partial x^{\alpha}}\otimes(L_{E_{H}}dx^{\beta})\right.\\
&+\left. (L_{E_{H}}[t,x^{\lambda}])\frac{\partial}{\partial x^{\alpha}}\otimes dt+[t,x^{\lambda}](L_{E_{H}}\frac{\partial}{\partial x^{\alpha}})\otimes dt+[t,x^{\lambda}]\frac{\partial}{\partial x^{\alpha}}\otimes(L_{E_{H}}dt)\right) \\
&=\epsilon^{\lambda\alpha}\left( \epsilon^{\mu\nu}\frac{\partial H}{\partial x^{\mu}}\frac{\partial[x^{\beta},x^{\lambda}]}{\partial x^{\nu}}\frac{\partial}{\partial x^{\alpha}}\otimes dx^{\beta}+\frac{\partial[x^{\beta},x^{\lambda}]}{\partial t}\frac{\partial}{\partial x^{\alpha}}\otimes dx^{\beta}-[x^{\beta},x^{\lambda}]\epsilon^{\mu\nu}\frac{\partial^{2}H}{\partial x^{\alpha}\partial x^{\mu}}\frac{\partial}{\partial x^{\nu}}\otimes dx^{\beta}\right. \\
&+\left. [x^{\beta},x^{\lambda}]\epsilon^{\mu\beta}\frac{\partial^{2}H}{\partial x^{\delta}\partial x^{\mu}}\frac{\partial}{\partial x^{\alpha}}\otimes dx^{\delta}+[x^{\beta},x^{\lambda}]\epsilon^{\mu\beta}\frac{\partial^{2}H}{\partial t\partial x^{\mu}}\frac{\partial}{\partial x^{\alpha}}\otimes dt+\epsilon^{\mu\nu}\frac{\partial H}{\partial x^{\mu}}\frac{\partial[t,x^{\lambda}]}{\partial x^{\nu}}\frac{\partial}{\partial x^{\alpha}}\otimes dt \right. \\
&+\left. \frac{\partial[t,x^{\lambda}]}{\partial t}\frac{\partial}{\partial x^{\alpha}}\otimes dt-[t,x^{\lambda}]\epsilon^{\mu\nu}\frac{\partial^{2}H}{\partial x^{\alpha}\partial x^{\mu}}\frac{\partial}{\partial x^{\nu}}\otimes dt \right)\\
&=\left(\epsilon^{\lambda\alpha}\epsilon^{\mu\nu}\frac{\partial H}{\partial x^{\mu}}\frac{\partial[x^{\beta},x^{\lambda}]}{\partial x^{\nu}}+\epsilon^{\lambda\alpha}\frac{\partial[x^{\beta},x^{\lambda}]}{\partial t}-\epsilon^{\lambda\nu}[x^{\beta},x^{\lambda}]\epsilon^{\mu\alpha}\frac{\partial^{2}H}{\partial x^{\nu}\partial x^{\mu}}+\epsilon^{\lambda\alpha}[x^{\delta},x^{\lambda}]\epsilon^{\mu\delta}\frac{\partial^{2}H}{\partial x^{\beta}\partial x^{\mu}}\right)\frac{\partial}{\partial x^{\alpha}}\otimes dx^{\beta}\\
&+\left( \epsilon^{\lambda\alpha}\epsilon^{\mu\nu}\frac{\partial H}{\partial x^{\mu}}\frac{\partial[t,x^{\lambda}]}{\partial x^{\nu}}+\epsilon^{\lambda\alpha}\frac{\partial[t,x^{\lambda}]}{\partial t}-\epsilon^{\lambda\nu}[t,x^{\lambda}]\epsilon^{\mu\alpha}\frac{\partial^{2}H}{\partial x^{\nu}\partial x^{\mu}}+\epsilon^{\lambda\alpha}[x^{\delta},x^{\lambda}]\epsilon^{\mu\delta}\frac{\partial^{2}H}{\partial t\partial x^{\mu}}\right) \frac{\partial}{\partial x^{\alpha}}\otimes dt.
\end{split}
\end{equation}
Let us remember the defining relation for $S$, $\overline{\Omega}=S\lrcorner \Omega$, so we have
\begin{equation}
\begin{split}
\overline{\Omega}=S\lrcorner \Omega &\Longrightarrow L_{E_{H}}\overline{\Omega}=L_{E_{H}}(S\lrcorner \Omega)\\
&\Longrightarrow -d(RK)\wedge\eta=(L_{E_{H}}S)\lrcorner\Omega+S\lrcorner(L_{E_{H}}\Omega)\\
&\Longrightarrow S\lrcorner(d(RH)\wedge\eta)-d(RK)\wedge\eta=(L_{E_{H}}S)\lrcorner\Omega.
\end{split}
\end{equation}
In short, the local form of $L_{E_{H}}S$ is
\begin{equation}
L_{E_{H}}S=A^{\alpha}_{\beta}\frac{\partial}{\partial x^{\alpha}}\otimes dx^{\beta}+A^{\alpha}_{t}\frac{\partial}{\partial x^{\alpha}}\otimes dt,
\end{equation}
hence
\begin{equation}
(S^{\alpha}_{\beta}\frac{\partial}{\partial x^{\alpha}}\otimes dx^{\beta}+S^{\alpha}_{t}\frac{\partial}{\partial x^{\alpha}}\otimes dt)\lrcorner\frac{\partial^{2}H}{\partial x^{\mu}\partial t}dx^{\mu}\wedge dt-\frac{\partial^{2}K}{\partial x^{\mu}\partial t}dx^{\mu}\wedge dt=(A^{\alpha}_{\beta}\frac{\partial}{\partial x^{\alpha}}\otimes dx^{\beta}+A^{\alpha}_{t}\frac{\partial}{\partial x^{\alpha}}\otimes dt)\lrcorner\frac{1}{2}\epsilon_{\mu\nu}dx^{\mu}\wedge dx^{\nu},
\end{equation}
which implies that
\begin{equation}
(S^{\alpha}_{\beta}\frac{\partial^{2}H}{\partial x^{\alpha}\partial t}-\frac{\partial^{2}K}{\partial x^{\beta}\partial t})dt=\epsilon_{\alpha\nu}A^{\alpha}_{\beta}dx^{\nu}\hspace{1cm}\textit{and}\hspace{1cm}S^{\alpha}_{t}\frac{\partial^{2}H}{\partial x^{\alpha}\partial t}dt+\frac{\partial^{2}K}{\partial x^{\nu}\partial t}dx^{\nu}=\epsilon_{\alpha\nu}A^{\alpha}_{t}dx^{\nu}
\end{equation}
as well as
\begin{equation}
A^{\alpha}_{\beta}=0\hspace{1cm}\textit{and}\hspace{1cm}A^{\alpha}_{t}=\epsilon^{\nu\alpha}\frac{\partial^{2}K}{\partial x^{\nu}\partial t}.
\end{equation}
So that the expression of $L_{E_{H}}S$ in the canonical coordinates $(x^{1},\cdots,x^{2n})=(q^{1},\cdots,q^{n},p_{1},\cdots,p_{n})$ is
\begin{equation}
L_{E_{H}}S=A^{\alpha}_{t}\frac{\partial}{\partial x^{\alpha}}\otimes dt=\epsilon^{\nu\alpha}\frac{\partial^{2}K}{\partial x^{\nu}\partial t}\frac{\partial}{\partial x^{\alpha}}\otimes dt.
\end{equation}
Therefore, for $l=1,2,\cdots$, we arrive to the following equation
\begin{equation}
L_{E_{H}}tr(S^{l})=tr(L_{E_{H}}S^{l})=\sum_{i=1}^{l-1}tr(S^{i-1}(L_{E_{H}}S)S^{l-i})=0.
\end{equation}
\end{proof}

Along the same lines as in the symplectic case, we can consider the Nijenhuis torsion tensor of $S$ and if $S$ is a Nijenhuis tensor then the traces of the powers of $S$ are pairwise in involution. In this case we have two Poisson brackets $\lbrace,\rbrace$ and $\overline{\lbrace,\rbrace}$ or two Poisson tensors $P$ and $\overline{P}$ defined by $\Omega$ and $\overline{\Omega}$ respectively. It is worth to mention that in this case the Nijenhuis tensor $S$ can not be defined by neither $P\overline{P}^{-1}$ nor by $\overline{P}P^{-1}$ since the Poisson tensors are not invertible (they are not defined by symplectic manifolds); despite this, if $S$ is a Nijenhuis tensor then the sequence of functions $(tr(S^{k}))_{k=1}^{\infty}$ satisfies the Lenard recursion relations for $P$ and $\overline{P}$. In canonical coordinates $(x^{1},\cdots,x^{2n},t)=(q^{1},\cdots,q^{n},p_{1},\cdots,p_{n},t)$ the Lenard recursion relations for the sequence of functions $(tr(S^{k}))_{k=1}^{\infty}$ have the form
\begin{equation}
\epsilon^{\lambda\alpha}\frac{\partial tr(S^{k})}{\partial x^{\alpha}}=\overline{\epsilon}^{\lambda\beta}\frac{\partial tr(S^{k+1})}{\partial x^{\beta}},
\end{equation}
where again $\overline{\epsilon}^{\lambda\beta}$ are the entries of the inverse matrix of $(\overline{\epsilon}_{\lambda\beta})=([x^{\lambda},x^{\beta}])$ (The components $[x^{\mu},t]$ of the 2-form $\Omega$ do not take part in the Poisson bracket); it is possible to consider the inverse of the matrix $(\overline{\epsilon}_{\lambda\beta})=([x^{\lambda},x^{\beta}])$ since the submanifolds of $M$ defined by $t=constant$ are the symplectic leaves of $\overline{P}$ (and also of $P$), indeed, locally we have
\begin{equation}
\overline{P}=\overline{\epsilon}^{\lambda\beta}\frac{\partial}{\partial x^{\lambda}}\otimes\frac{\partial}{\partial x^{\beta}}. 
\end{equation}
(For details about Poisson tensors see \cite{PS}). Since the components $S_{t}^{\alpha}$ of the tensor $S$ are not involved in the traces of the powers of $S$ then, similarly to the symplectic case, we have
\begin{equation} 
N_{\beta\gamma}^{\lambda}(S^{k-1})_{\lambda}^{\gamma}=S^{\alpha}_{\beta}\frac{\partial tr(S^{k})}{\partial x^{\alpha}}-\frac{\partial tr(S^{k+1})}{\partial x^{\beta}}.
\end{equation}
Hence, we conclude that if the Nijenhuis torsion tensor of $S$ vanishes then the sequence of functions $(tr(S^{k}))_{k=1}^{\infty}$ satisfies the Lenard recursion relations for $P$ and $\overline{P}$ and, consequently, the functions $tr(S^{k})$ are pairwise in involution with respect to both Poisson brackets.

\section{Contact geometry and contact Hamiltonian systems}
\label{sec6}

In this section we move to the formalism of contact Hamiltonian systems which naturally describe dissipative systems (for details see \cite{Contact,CoCo,Brav}), and they also have applications in thermodynamics, statistical mechanics and others \cite{CHD}. 

\begin{de}
Let $M$ be a $2n+1$ dimensional smooth manifold. A contact structure on $M$ is a 1-form $\theta$ on $M$ such that $\theta\wedge d\theta^{n}\neq 0$. If $\theta$ is a contact structure on $M$ we say that $(M,\theta)$ is a contact manifold.
\end{de}

Let $(M,\theta)$ be a contact manifold of dimension $2n+1$. Around any point $p\in M$ there exist local coordinates $(q^{1},\cdots,q^{n},p_{1},\cdots,p_{n},z)$, i.e. Darboux coordinates, such that
\begin{equation}
\theta=dz-p_{i}dq^{i}.
\end{equation}
There exists a distinguished vector field $R$ on $M$, called the Reeb vector field, which obeys 
\begin{equation}
R\lrcorner \theta =1 \hspace{1cm}\textit{and}\hspace{1cm} R\lrcorner d\theta =0.
\end{equation}
In canonical coordinates we simply have $R=\frac{\partial}{\partial z}$.

For each $f\in C^{\infty}(M)$ is assigned a vector field $X_{f}$ on $M$, called the Hamiltonian vector field for $f$, according to
\begin{equation}
X_{f}\lrcorner \theta = -f \hspace{1cm}\textit{and}\hspace{1cm} X_{f}\lrcorner d\theta =df-R(f)\theta.
\end{equation}
In canonical coordinates we have
\begin{equation}
X_{f}=\frac{\partial f}{\partial p_{i}}\frac{\partial}{\partial q^{i}}-\left( \frac{\partial f}{\partial q^{i}}+p_{i}\frac{\partial f}{\partial z}\right) \frac{\partial}{\partial p_{i}}+\left( p_{i}\frac{\partial f}{\partial p_{i}}-f\right)\frac{\partial}{\partial z}.
\end{equation}
It can be checked that the assignment $f\longmapsto X_{f}$ is linear, that is
\begin{equation}
X_{f+\alpha g}=X_{f}+\alpha X_{g},
\end{equation}
$\forall f,g\in C^{\infty}(M)$ and $\forall \alpha \in\mathbb{R}$.

Symplectic and cosymplectic manifolds are Poisson manifolds (symplectic and cosymplectic structures define Poisson brackets), but a contact manifold is strictly a Jacobi manifold, i.e., a contact structure on a manifold defines a Jacobi bracket. Given $f,g \in C^{\infty}(M)$ the Jacobi bracket of $f$ and $g$ is defined by
\begin{equation}
\lbrace f,g\rbrace=X_{g}(f)+fR(g)=\theta([X_{f},X_{g}]).
\end{equation}
In canonical coordinates we have
\begin{equation}
\lbrace f,g\rbrace=\frac{\partial f}{\partial q^{i}}\frac{\partial g}{\partial p_{i}}-\frac{\partial f}{\partial p_{i}}\frac{\partial g}{\partial q^{i}}+\frac{\partial f}{\partial z}\left( p_{i}\frac{\partial g}{\partial p_{i}}-g\right)-\frac{\partial g}{\partial z}\left( p_{i}\frac{\partial f}{\partial p_{i}}-f\right).
\end{equation}

Hamiltonian systems on contact manifolds are called contact Hamiltonian systems. Given $H\in C^{\infty}(M)$ the dynamics of the Hamiltonian on $(M,\theta)$ (the phase space) with Hamiltonian function $H$ is defined by the Hamiltonian vector field $X_{H}$.
In canonical coordinates we have
\begin{equation}
X_{H}=\frac{\partial H}{\partial p_{i}}\frac{\partial}{\partial q^{i}}-\left( \frac{\partial H}{\partial q^{i}}+p_{i}\frac{\partial H}{\partial z}\right) \frac{\partial}{\partial p_{i}}+\left( p_{i}\frac{\partial H}{\partial p_{i}}-H\right)\frac{\partial}{\partial z}.
\end{equation}
The trajectories $\psi(t)=(q^{1}(t),\cdots,q^{n}(t),p_{1}(t),\cdots,p_{n}(t),z(t))$ of the system are the integral curves of $X_{H}$, they satisfy the dissipative Hamilton equations of motion
\begin{equation}
\dot{q^{i}} =\frac{\partial H}{\partial p_{i}}, \hspace{1cm}
\dot{p_{i}} =-\left( \frac{\partial H}{\partial q^{i}}+p_{i}\frac{\partial H}{\partial z}\right) , \hspace{1cm}
\dot{z}=p_{i}\frac{\partial H}{\partial p_{i}}-H.
\end{equation}

The evolution of a function $f\in C^{\infty}(M)$ (an observable) along the trajectories of the system reads
\begin{equation}
\dot{f}= \mathcal{L}_{X_{H}}f=X_{H}f=\lbrace f,H\rbrace-fR(H).
\end{equation}
We say that a function $f\in C^{\infty}(M)$ is a constant of motion of the system if it is constant along the trajectories of the system, that is, $f$ is a constant of motion if $L_{X_{H}}f=0$ ($\lbrace f,H\rbrace-fR(H)=0$).

\section{Canonical and canonoid transformations for contact Hamiltonian systems}
\label{sec7}

By following the ideas developed in the  previous sections, in this section we introduce the notion of canonoid transformations for contact Hamiltonian systems. As in sections \ref{sec3} and \ref{sec5}, the main result will show that having a canonoid transformation for a given contact Hamiltonian system one can find constants of motion (conserved quantities).

Let us consider a contact manifold $(M,\theta)$ of dimension $2n+1$.
\begin{de}
We say that a diffeomorphism $F:M\longrightarrow M$ is a canonical transformation if $F^{*}\theta=\theta$.
\end{de}
Let $F:M\longrightarrow M$ be a diffeomorphism. We know that $F^{*}\theta$ is a contact structure on $M$, let us denote it by $\overline{\theta}$, by $\overline{\lbrace,\rbrace}$ the Jacobi bracket defined by it and by $\overline{R}$ the Reeb vector field for this contact structure. Around any point $p\in M$ there are local coordinates $(Q^{1},\cdots,Q^{n},P_{1},\cdots,P_{n},Z)$ such that
\begin{equation}
\overline{\theta}=dZ-P_{i}dQ^{i}\hspace{1cm}\hspace{1cm}\textit{and}\hspace{1cm}\overline{R}=\frac{\partial}{\partial Z}.
\end{equation}
Given $f\in C^{\infty}(M)$, let us denote by $\overline{X}_{f}$ the Hamiltonian vector field for $f$ defined by $\overline{\theta}$, then $\overline{\lbrace f,g\rbrace}=\overline{X}_{g}f+f\overline{R}(g)$. The following propositions are equivalent
\begin{enumerate}
\item $F$ is a canonical transformation,
\item $\overline{X}_{f}=X_{f}$ $\forall f \in C^{\infty}(M)$ and
\item $\overline{\lbrace f,g\rbrace}=\lbrace f,g\rbrace$ $\forall f,g \in C^{\infty}(M)$.
\end{enumerate}
So, a diffeomorphism from a cosymplectic manifold to itself is a canonical transformation if and only if it preserves the Jacobi bracket. 

Let us consider the contact Hamiltonian system $(M,\theta,H)$; if $F$ is a canonical transformation then the equations of motion in the coordinates $(Q^{1},\cdots,Q^{n},P_{1},\cdots,P_{n},Z)$ are
\begin{equation}
\begin{split}
\dot{Q^{i}} &= \lbrace Q^{i},H\rbrace-Q^{i}R(H)=\overline{\lbrace Q^{i},H\rbrace}-Q^{i}\overline{R}(H)=\frac{\partial H}{\partial P_{i}},\\
\dot{P_{i}} &=\lbrace P_{i},H\rbrace-P_{i}R(H)=\overline{\lbrace P_{i},H\rbrace}-P_{i}\overline{R}(H)=-\left( \frac{\partial H}{\partial Q^{i}}+P_{i}\frac{\partial H}{\partial Z}\right) ,\\
\dot{Z}&=\lbrace Z,H\rbrace-ZR(H)=\overline{\lbrace Z,H\rbrace}-Z\overline{R}(H)=P_{i}\frac{\partial H}{\partial P_{i}}-H.
\end{split}
\end{equation}
So that canonical transformations preserve the Hamiltonian form of the equations of motion. As in the symplectic and cosymplectic case, there are transformations for contact Hamiltonian systems that preserve the Hamiltonian form of equations of motion and they are not necessarily canonical, we also called canonoid transformations these more general transformations for contact Hamiltonian systems.

\begin{de}
We say that the diffeomorphism $F:M\longrightarrow M$ is a canonoid transformation for the contact Hamiltonian system $(M,\theta,H)$ if there exists a function $K\in C^{\infty}(M)$ such that 
\begin{equation}
X_{H}\lrcorner \overline{\theta} = -K \hspace{1cm}\textit{and}\hspace{1cm} X_{H}\lrcorner d\overline{\theta} =dK-\overline{R}(K)\overline{\theta}.
\end{equation}
\end{de}
Observe that every canonical transformation is also a canonoid transformation.

Let $F$ be a canonoid transformation. Given $f\in C^{\infty}(M)$ we have that
\begin{equation}
\dot{f}=\lbrace f,H\rbrace-fR(H)=X_{H}f+fR(H)-fR(H)=\overline{\lbrace f,K\rbrace}-f\overline{R}(K),
\end{equation}
so that 
\begin{equation}
\lbrace f,H\rbrace=\overline{\lbrace f,K\rbrace}-f\overline{R}(K)+fR(H).
\end{equation}
The equations of motion in the coordinates $(Q^{1},\cdots,Q^{n},P_{1},\cdots,P_{n},Z)$ are
\begin{equation}
\begin{split}
\dot{Q^{i}} &=\overline{\lbrace Q^{i},K\rbrace}-Q^{i}R(H)=\frac{\partial K}{\partial P_{i}},\\
\dot{P_{i}} &=\overline{\lbrace P_{i},K\rbrace}-P_{i}R(H)=-\left( \frac{\partial K}{\partial Q^{i}}+P_{i}\frac{\partial K}{\partial Z}\right),\\
\dot{Z} &=\overline{\lbrace Z,K\rbrace}-ZR(H)=P_{i}\frac{\partial K}{\partial P_{i}}-K.
\end{split}
\end{equation}
So that canonoid transformations preserve the Hamiltonian form of the equations of motion. We can observe that, reciprocally, transformations that preserve the Hamiltonian form of the equations of motion are canonoid.

As in the symplectic and cosymplectic case, we have two Hamiltonian systems $(M,\theta,H)$ and $(M,\overline{\theta}=F^{*}\theta,K)$ with the same phase space and the dynamics defined by the same Hamiltonian vector field $X_{H}$ (or two contact Hamiltonian representations for the same mechanical system). As in the cosymplectic case, we have that the geometric structures $\theta$ and $\overline{\theta}$ are not invariant under the flow of $X_{H}$, namely
\begin{equation}
L_{X_{H}}\theta=-R(H)\theta \hspace{1cm}\textit{and}\hspace{1cm}L_{X_{H}}\overline{\theta}=-\overline{R}(K)\overline{\theta}.
\end{equation}
Following the ideas presented in the previous sections, we define the (1,1)-tensor field $S$ on $M$ by
\begin{equation}
d\overline{\theta}=S\lrcorner d\theta \hspace{1cm}\textit{and}\hspace{1cm}S\lrcorner\theta=0.
\end{equation}
Let us write down $S$ in canonical coordinates $(q^{1},\cdots,q^{n},p_{1},\cdots,p_{n},z)$. Next, in order to have a compact local form of the involved geometric objects, we write $(x^{1},\cdots,x^{2n})=(q^{1},\cdots,q^{n},p_{1},\cdots,p_{n})$ then we have
\begin{equation}
d\theta=dq^{i}\wedge dp_{i}=\frac{1}{2}\epsilon_{\mu\nu}dx^{\mu}\wedge dx^{\nu} \ ,
\end{equation}
and 
\begin{equation}
d\overline{\theta}=dQ^{i}\wedge dP_{i}=\frac{1}{2}[x^{\mu},x^{\nu}]dx^{\mu}\wedge dx^{\nu}+[x^{\mu},z]dx^{\mu}\wedge dz.
\end{equation}
The local form of the tensor $S$ is
\begin{equation}
S=S^{\alpha}_{\beta}\frac{\partial}{\partial x^{\alpha}}\otimes dx^{\beta}+S^{\alpha}_{z}\frac{\partial}{\partial x^{\alpha}}\otimes dz+S^{z}_{\beta}\frac{\partial}{\partial z}\otimes dx^{\beta}+S^{z}_{z}\frac{\partial}{\partial z}\otimes dz.
\end{equation}
Since $S\lrcorner\theta=0$ then $S^{z}_{\beta}=p_{i}S^{i}_{\beta}$ and $S^{z}_{z}=p_{i}S^{i}_{z}$.
By the same calculus presented in section \ref{sec5} we conclude that
\begin{equation}
S^{\alpha}_{\beta}=\epsilon^{\lambda\alpha}[x^{\beta},x^{\lambda}]\hspace{1cm}\textit{and}\hspace{1cm}S^{\alpha}_{z}=\epsilon^{\lambda\alpha}[z,x^{\lambda}];
\end{equation}
so we have that the expression of $S$ in the canonical coordinates $(x^{1},\cdots,x^{2n},z)=(q^{1},\cdots,q^{n},p_{1},\cdots,p_{n},z)$ is
\begin{equation}
\label{Scont}
S=\epsilon^{\lambda\alpha}[x^{\beta},x^{\lambda}]\frac{\partial}{\partial x^{\alpha}}\otimes dx^{\beta}+\epsilon^{\lambda\alpha}[z,x^{\lambda}]\frac{\partial}{\partial x^{\alpha}}\otimes dz+p_{i}\epsilon^{\lambda i}[x^{\beta},x^{\lambda}]\frac{\partial}{\partial z}\otimes dx^{\beta}+p_{i}\epsilon^{\lambda i}[z,x^{\lambda}]\frac{\partial}{\partial z}\otimes dz.
\end{equation}

The notion of integrability for contact Hamiltonian systems is a bit different from that one for time-dependent or time-independent Hamiltonian systems. In particular, there are important geometric differences between the phase spaces of contact Hamiltonian systems and time-dependent or time-independent Hamiltonian systems. For the first ones the phase spaces are strictly Jacobi manifolds and for the second ones they are Poisson manifolds. The following definition is presented in \cite{Boyer}:
\begin{de}
The contact Hamiltonian system $(M,\theta,H)$ is said to be completely integrable if there exists $n+1$ independent constants of motion $H,f_{1},\cdots,f_{n}$ in involution.
\end{de}
According to this definition, for a contact Hamiltonian system to be integrable it is necessary that the Hamiltonian function is a constant of motion, which is not true in general, so we need to restrict our class of Hamiltonian systems. The following definition is presented in \cite{Boyer}:
\begin{de}
We say that the contact Hamiltonian system $(M,\theta,H)$ is good if $H$ is constant along the flow of the Reeb vector field $R$, or equivalently $H$ is a constant of motion. 
\end{de}
One important kind of a good contact Hamiltonian system is given by $(M,\theta,1)$. In this case, the dynamics of the system is defined by the Reeb vector field $R$; these kind of systems are called of Reeb type \cite{Visi}.

Let us suppose that $(M,\theta,H)$ and $(M,\overline{\theta},K)$ are good contact Hamiltonian systems, that is, $R(H)=\overline{R}(K)=0$. Then as in the symplectic case we have that the structures $\theta$ and $\overline{\theta}$ are invariant under the flow of $X_{H}$, indeed
\begin{equation}
L_{X_{H}}\theta=-R(H)\theta=0 \hspace{1cm}\textit{and}\hspace{1cm} L_{X_{H}}\overline{\theta}=-\overline{R}(H)\overline{\theta}=0,
\end{equation}
which implies that $L_{X_{H}}d\theta=L_{X_{H}}d\overline{\theta}=0$, so by following the same approach presented in section \ref{sec3} we have that the traces of the powers of $S$ are constants of motion.

\begin{te}\label{te3}
If $F:M\longrightarrow M$ is a canonoid transformation for the good contact Hamiltonian system $(M,\theta,H)$, then the traces of the powers of the (1,1)-tensor field $S$ defined above in (\ref{Scont}) are constants of motion.
\end{te}

\begin{proof}
\begin{equation}
\begin{split}
L_{X_{H}}d\overline{\theta}=(L_{X_{H}}S)\lrcorner d\theta+S\lrcorner(L_{X_{H}}d\theta)&\Longrightarrow  0=(L_{X_{H}}S)\lrcorner d\theta\\
&\Longrightarrow L_{X_{H}}S=0. 
\end{split}
\end{equation}
Then for $l=1,2,\cdots$, we have
\begin{equation}
L_{X_{H}}tr(S^{l})=tr(L_{X_{H}}S^{l})=\sum_{i=1}^{l-1}tr(S^{i-1}(L_{X_{H}}S)S^{l-i})=0.
\end{equation}
\end{proof}

\section{Cocontact geometry and time-dependent contact Hamiltonian systems}
\label{sec8}

This section together with sections \ref{sec2}, \ref{sec4} and \ref{sec6} completes the notation and language employed in this paper. Here we briefly review the formalism of time-dependent contact Hamiltonian systems introduced in \cite{Cocontact}.

\begin{de}
Let $M$ be a $2n+2$ dimensional smooth manifold. A cocontact structure on $M$ is a couple $(\theta,\eta)$ of 1-forms on $M$ such that $\eta$ is closed and $\eta\wedge\theta\wedge(d\theta)^{n}\neq 0$. If $(\theta,\eta)$ is a cocontact structure on $M$ we say that $(M,\theta,\eta)$ is a cocontact manifold.
\end{de}

Let $(M,\theta,\eta)$ be a cocontact manifold of dimension $2n+2$. Around any point $p\in M$ there exist local coordinates $(t,q^{1},\cdots,q^{n},p_{1},\cdots,p_{n},z)$, called canonical coordinates or Darboux coordinates, such that
\begin{equation}
\theta=dz-p_{i}dq^{i}\hspace{1cm}\textit{and}\hspace{1cm}\eta=dt.
\end{equation}
There exist two distinguished vector fields $R_{z}$ and $R_{t}$ on $M$, called the contact Reeb vector field and the time Reeb vector field, respectively, such that 
\begin{equation}
\left\lbrace \begin{array}{c}
\eta(R_{z})=0 \\ 
\theta(R_{z})=1 \\ 
d\theta(R_{z})=0
\end{array} \right. 
\end{equation}
and
\begin{equation}
\left\lbrace \begin{array}{c}
\eta(R_{t})=1 \\ 
\theta(R_{t})=0 \\ 
d\theta(R_{t})=0
\end{array} \right. .
\end{equation}
In canonical coordinates, we have $R_{z}=\frac{\partial}{\partial z}$ and $R_{t}=\frac{\partial}{\partial t}$.

For each $f\in C^{\infty}(M)$ is assigned a vector field $X_{f}$ on $M$, called the Hamiltonian vector field for $f$, according to
\begin{equation}
X_{f}\lrcorner \theta=-f,\hspace{1cm}X_{f}\lrcorner d\theta=df-R_{z}(f)\theta-R_{t}(f)\eta \hspace{1cm}\textit{and}\hspace{1cm} X_{f}\lrcorner \eta =0.
\end{equation}
In canonical coordinates, we have
\begin{equation}
X_{f}=\frac{\partial f}{\partial p_{i}}\frac{\partial}{\partial q^{i}}-\left( \frac{\partial f}{\partial q^{i}}+p_{i}\frac{\partial f}{\partial z}\right) \frac{\partial}{\partial p_{i}}+\left( p_{i}\frac{\partial f}{\partial p_{i}}-f\right)\frac{\partial}{\partial z}.
\end{equation}
We can observe that the assignment $f\longmapsto X_{f}$ is linear, that is
\begin{equation}
X_{f+\alpha g}=X_{f}+\alpha X_{g},
\end{equation}
$\forall f,g\in C^{\infty}(M)$ and $\forall \alpha \in\mathbb{R}$. Like contact manifolds, cocontact manifolds are Jacobi manifolds; given $f,g \in C^{\infty}(M)$ the Jacobi bracket of $f$ and $g$ is defined by
\begin{equation}
\lbrace f,g\rbrace=X_{g}(f)+fR_{z}(g).
\end{equation}
In canonical coordinates, we have
\begin{equation}
\lbrace f,g\rbrace=\frac{\partial f}{\partial q^{i}}\frac{\partial g}{\partial p_{i}}-\frac{\partial f}{\partial p_{i}}\frac{\partial g}{\partial q^{i}}+\frac{\partial f}{\partial z}\left( p_{i}\frac{\partial g}{\partial p_{i}}-g\right)-\frac{\partial g}{\partial z}\left( p_{i}\frac{\partial f}{\partial p_{i}}-f\right).
\end{equation}

Time-dependent contact Hamiltonian systems are defined under the formalism of cocontact geometry. Given $H\in C^{\infty}(M)$ the dynamic of the Hamiltonian system on $(M,\theta,\eta)$ (the phase space) with Hamiltonian function $H$ is defined by the evolution vector field $E_{H}=X_{H}+R_{t}$.
In canonical coordinates
\begin{equation}
E_{H}=\frac{\partial H}{\partial p_{i}}\frac{\partial}{\partial q^{i}}-\left( \frac{\partial H}{\partial q^{i}}+p_{i}\frac{\partial H}{\partial z}\right) \frac{\partial}{\partial p_{i}}+\left( p_{i}\frac{\partial H}{\partial p_{i}}-H\right)\frac{\partial}{\partial z}+\frac{\partial}{\partial t}.
\end{equation}
The trajectories $\psi(s)=(t(s),q^{1}(s),\cdots,q^{n}(s),p_{1}(s),\cdots,p_{n}(s),z(s))$ of the system are the integral curves of $E_{H}$, they satisfy the equations
\begin{equation}
\dot{q^{i}} =\frac{\partial H}{\partial p_{i}}, \hspace{1cm}
\dot{p_{i}} =-\left( \frac{\partial H}{\partial q^{i}}+p_{i}\frac{\partial H}{\partial z}\right) , \hspace{1cm}
\dot{z}=p_{i}\frac{\partial H}{\partial p_{i}}-H, \hspace{1cm} \dot{t}=1.
\end{equation}
Since $\dot{t}=1$ then $t=s$, which implies that the temporal parameter for the system is $t$, that is, the trajectories of the system are parametrized by $t$
\begin{equation}
\psi(t)=(t,q^{1}(t),\cdots,q^{n}(t),p_{1}(t),\cdots,p_{n}(s),z(t))
\end{equation}
and they obey the dissipative Hamilton equations of motion 
\begin{equation}
\dot{q^{i}} =\frac{\partial H}{\partial p_{i}}, \hspace{1cm}
\dot{p_{i}} =-\left( \frac{\partial H}{\partial q^{i}}+p_{i}\frac{\partial H}{\partial z}\right) , \hspace{1cm}
\dot{z}=p_{i}\frac{\partial H}{\partial p_{i}}-H.
\end{equation}
The evolution of a function $f\in C^{\infty}(M)$ (an observable) along the trajectories of the system is given by
\begin{equation}
\dot{f}= \mathcal{L}_{E_{H}}f=E_{H}f=X_{H}f+R_{t}f=\lbrace f,H\rbrace-fR_{z}H+R_{t}f.
\end{equation}
We say that a function $f\in C^{\infty}(M)$ is a constant of motion of the system if it is constant along the trajectories of the system, that is, $f$ is a constant of motion if $L_{E_{H}}f=0$ ($\lbrace f,H\rbrace-fR_{z}H+R_{t}f=0$).

\section{Canonical and canonoid transformations for time-dependent contact Hamiltonian systems}
\label{sec9}

Finally, in this section we introduce the notions of canonical and canonoid transformations for time-dependent contact Hamiltonian systems and we obtain analogous results to those presented in section \ref{sec7}.

Let us consider a cocontact manifold $(M,\theta,\eta)$ of dimension $2n+2$.
\begin{de}
We say that a diffeomorphism $F:M\longrightarrow M$ is a canonical transformation if $F^{*}\theta=\theta$ and $F^{*}\eta=\eta$.
\end{de}
Let $F:M\longrightarrow M$ be a diffeomorphism. We know that $(F^{*}\theta,F^{*}\eta)$ is a cocontact structure on $M$, let us denote it by $(\overline{\theta},\overline{\eta})$, by $\overline{\lbrace,\rbrace}$ the Jacobi bracket defined by it and by $\overline{R}_{z},\overline{R}_{t}$ the contact and time Reeb vector field respectively for this contact structure. Around any point $p\in M$ there are local coordinates $(T,Q^{1},\cdots,Q^{n},P_{1},\cdots,P_{n},Z)$ such that
\begin{equation}
\overline{\theta}=dZ-P_{i}dQ^{i}\hspace{1cm}\hspace{1cm}\textit{and}\hspace{1cm}\overline{R}=\frac{\partial}{\partial Z} \hspace{1cm}\textit{and}\hspace{1cm}\overline{\eta}=dT.
\end{equation}
Given $f\in C^{\infty}(M)$, let us denote by $\overline{X}_{f}$ the Hamiltonian vector field for $f$ defined by $(\overline{\theta},\overline{\eta})$, then $\overline{\lbrace f,g\rbrace}=\overline{X}_{g}f+f\overline{R}_{z}(g)$.

As in the contact case we have that canonical coordinates preserve the Jacobi bracket and the Hamiltonian form of the equations of motion (we avoid to write the whole details since it is repetitive).

Now, we introduce canonoid transformations.

\begin{de}
We say that the diffeomorphism $F:M\longrightarrow M$ is a canonoid transformation for the time-dependent contact Hamiltonian system $(M,\theta,\eta,H)$ if there exists a function $K\in C^{\infty}(M)$ such that 
\begin{equation}
X_{H}\lrcorner \overline{\theta} = -K, \hspace{1cm} X_{H}\lrcorner d\overline{\theta} =dK-\overline{R}_{z}(K)\overline{\theta}-\overline{R}_{t}(K)\overline{\eta}\hspace{1cm}\textit{and}\hspace{1cm}X_{H}\lrcorner \overline{\eta}=0.
\end{equation}
\end{de}

Let $F$ be a canonoid transformation. We have two time-dependent contact Hamiltonian systems $(M,\theta,\eta,H)$ and $(M,\overline{\theta},\overline{\eta},K)$; the temporal parameter for the first one is $t$ and the temporal parameter for the second one is $T$, then $T=t$ (formally $T=F^{*}t$), so we have that $\overline{\eta}=\eta$ and $\overline{R}_{t}=R_{t}$. The dynamic of both systems is defined by the same evolution vector field $E_{H}$ (so that again we have two Hamiltonian representations for the same mechanical system).

Given $f\in C^{\infty}(M)$ we have that
\begin{equation}
\dot{f}=\lbrace f,H\rbrace-fR_{z}H+R_{t}f=X_{H}f+fR_{z}H-fR_{z}H+R_{t}f=\overline{\lbrace f,K\rbrace}-f\overline{R}_{z}K+R_{t}f,
\end{equation}
then $\lbrace f,H\rbrace=\overline{\lbrace f,K\rbrace}-f\overline{R}_{z}K+fR_{z}H$ (in general, canonoid transformations do not preserve the Poisson bracket).
The equations of motion in the coordinates $(t,Q^{1},\cdots,Q^{n},P_{1},\cdots,P_{n},Z)$ take the form
\begin{equation}
\begin{split}
\dot{Q^{i}} &= \overline{\lbrace Q^{i},K\rbrace}-Q^{i}\overline{R}_{z}K+R_{t}Q^{i}=\frac{\partial K}{\partial P_{i}},\\
\dot{P_{i}} &=-\left( \frac{\partial K}{\partial Q^{i}}+P_{i}\frac{\partial K}{\partial Z}\right) ,\\
\dot{Z} &=P_{i}\frac{\partial K}{\partial P_{i}}-K.
\end{split}
\end{equation}
So that canonoid transformations preserve the Hamiltonian form of the equations of motion.

As in the cosymplectic case we have that the geometric structures $\theta$ and $\overline{\theta}$ are not invariant under the flow of $E_{H}$, namely
\begin{equation}
L_{E_{H}}\theta=E_{H}\lrcorner d\theta+d(E_{H}\lrcorner \theta)=(X_{H}+R_{t})\lrcorner d\theta+d((X_{H}+R_{t})\lrcorner \theta)=dH-R_{z}(H)\theta-R_{t}(H)\eta+d(-H)=-R_{z}(H)\theta-R_{t}(H)\eta
\end{equation}
and analogously we see that 
\begin{equation}
L_{E_{H}}\overline{\theta}=-\overline{R}_{z}(K)\overline{\theta}-R_{t}(K)\eta.
\end{equation}
Following the development in the cosymplectic case, we define the (1,1)-tensor field $S$ on $M$ by
\begin{equation}
d\overline{\theta}=S\lrcorner d\theta,\hspace{1cm}S\lrcorner\theta=0 \hspace{1cm}\textit{and}\hspace{1cm}S\lrcorner\eta=0.
\end{equation}
 Let us see $S$ in canonical coordinates $(t,q^{1},\cdots,q^{n},p_{1},\cdots,p_{n},z)$. Again we write $(x^{1},\cdots,x^{2n})=(q^{1},\cdots,q^{n},p_{1},\cdots,p_{n})$ then we have
\begin{equation}
d\theta=dq^{i}\wedge dp_{i}=\frac{1}{2}\epsilon_{\mu\nu}dx^{\mu}\wedge dx^{\nu}
\end{equation}
and 
\begin{equation}
d\overline{\theta}=dQ^{i}\wedge dP_{i}=\frac{1}{2}[x^{\mu},x^{\nu}]dx^{\mu}\wedge dx^{\nu}+[x^{\mu},z]dx^{\mu}\wedge dz+[x^{\mu},t]dx^{\mu}\wedge dt+[z,t]dz\wedge dt.
\end{equation}
The local form of the tensor $S$ is
\begin{equation}
S=S^{\alpha}_{\beta}\frac{\partial}{\partial x^{\alpha}}\otimes dx^{\beta}+S^{\alpha}_{z}\frac{\partial}{\partial x^{\alpha}}\otimes dz+S^{z}_{\beta}\frac{\partial}{\partial z}\otimes dx^{\beta}+S^{z}_{z}\frac{\partial}{\partial z}\otimes dz+S^{\alpha}_{t}\frac{\partial}{\partial x^{\alpha}}\otimes dt+S^{t}_{\beta}\frac{\partial}{\partial t}\otimes dx^{\beta}+S^{t}_{t}\frac{\partial}{\partial t}\otimes dt+S^{z}_{t}\frac{\partial}{\partial z}\otimes dt+S^{t}_{z}\frac{\partial}{\partial t}\otimes dz.
\end{equation}
Since $S\lrcorner\eta=0$ we have $S^{t}_{\beta}=S^{t}_{t}=S^{t}_{z}=0$ and since $S\lrcorner\theta=0$ then $S^{z}_{\beta}=p_{i}S^{i}_{\beta}$, $S^{z}_{z}=p_{i}S^{i}_{z}$ and $S^{z}_{t}=p_{i}S^{i}_{t}$.
By the same calculus presented in section \ref{sec5} we have that
\begin{equation}
S^{\alpha}_{\beta}=\epsilon^{\lambda\alpha}[x^{\beta},x^{\lambda}],\hspace{1cm}S^{\alpha}_{z}=\epsilon^{\lambda\alpha}[z,x^{\lambda}]\hspace{1cm}\textit{and}\hspace{1cm}S^{\alpha}_{t}=\epsilon^{\lambda \alpha}[t,x^{\lambda}];
\end{equation}
so we have that the expression of $S$ in the canonical coordinates $(t,x^{1},\cdots,x^{2n},z)=(t,q^{1},\cdots,q^{n},p_{1},\cdots,p_{n},z)$ is given by
\begin{equation}
\label{Scoco}
\begin{split}
S &=\epsilon^{\lambda\alpha}[x^{\beta},x^{\lambda}]\frac{\partial}{\partial x^{\alpha}}\otimes dx^{\beta}+\epsilon^{\lambda\alpha}[z,x^{\lambda}]\frac{\partial}{\partial x^{\alpha}}\otimes dz+p_{i}\epsilon^{\lambda i}[x^{\beta},x^{\lambda}]\frac{\partial}{\partial z}\otimes dx^{\beta}+p_{i}\epsilon^{\lambda i}[z,x^{\lambda}]\frac{\partial}{\partial z}\otimes dz\\
 &+\epsilon^{\lambda \alpha}[t,x^{\lambda}]\frac{\partial}{\partial x^{\alpha}}\otimes dt+p_{i}\epsilon^{\lambda i}[t,x^{\lambda}]\frac{\partial}{\partial z}\otimes dt.
\end{split}
\end{equation}

Let us remember from section \ref{sec7} that a contact Hamiltonian system with Hamiltonian function $f_{0}$ and phase space of dimension $2n+1$ is said to be completely integrable if there exists $n+1$ independent constants of motion $f_{0},f_{1},\cdots,f_{n}$ in involution; which means that for integrability it is essential that the Hamiltonian function is a constant of motion. Thus, in analogy with the good contact time-independent  Hamiltonian system presented previously, we propose the following definition: 
\begin{de}
We say that the time-dependent contact Hamiltonian system $(M,\theta,\eta,H)$ is good if $H$ is constant along the flow of the contact and Reeb vector fields $R_{z}$ and $R_{t}$, respectively, or equivalently $H$ is a constant of motion. 
\end{de}
Let us suppose that $(M,\theta,\eta,H)$ and $(M,\overline{\theta},\eta,K)$ are good contact Hamiltonian systems, that is, $R_{z}(H)=\overline{R}_{z}(K)=R_{t}(H)=R_{t}(K)=0$. Then the structures $\theta$ and $\overline{\theta}$ are invariant under the flow of $E_{H}$, which implies that $L_{E_{H}}d\theta=L_{E_{H}}d\overline{\theta}=0$, so by following the same development presented in section \ref{sec7} we have that the traces of the powers of $S$ are constants of motion.

\begin{te}\label{te4}
If $F:M\longrightarrow M$ is a canonoid transformation for the good time-dependent contact Hamiltonian system $(M,\theta,\eta,H)$, then the traces of the powers of the (1,1)-tensor field $S$ defined above in (\ref{Scoco}) are constants of motion.
\end{te}

\begin{proof}
\begin{equation}
\begin{split}
L_{E_{H}}d\overline{\theta}=(L_{E_{H}}S)\lrcorner d\theta+S\lrcorner(L_{E_{H}}d\theta)&\Longrightarrow  0=(L_{E_{H}}S)\lrcorner d\theta\\
&\Longrightarrow L_{E_{H}}S=0. 
\end{split}
\end{equation}
Then for $l=1,2,\cdots$, we have
\begin{equation}
L_{E_{H}}tr(S^{l})=tr(L_{E_{H}}S^{l})=\sum_{i=1}^{l-1}tr(S^{i-1}(L_{E_{H}}S)S^{l-i})=0.
\end{equation}
\end{proof}

\section{Conclusions}

In this study we have presented canonoid transformations from a geometric view point. These transformations play a significant role in the theory of Hamiltonian systems. In particular, for time-independent and time-dependent Hamiltonian systems we have presented the geometrically-based concept of canonoid transformation under the formalisms of symplectic and cosymplectic geometry, respectively. In the case of (time-independent) contact Hamiltonian systems and (time-dependent) cocontact Hamiltonian systems, the notion of canonoid transformation was depicted in a natural way following the approach used on the symplectic and cosymplectic cases. In all these cases, we have shown not only the existence of constants of motion associated with canonoid transformations, but also a constructive process to find them was described. We hope this work will be helpful in shedding light on the lesser-known aspects and few explored theories of contact and cocontact Hamiltonian systems. In future works, we plan to explore the analogue of canonoid transformations in quantum mechanics.

\section*{Acknowledgments}

The author R. Azuaje wishes to thank CONACYT (México) for financial support through a postdoctoral fellowship in the program Estancias Posdoctorales por México 2022.

\end{document}